\documentclass[conference]{IEEEtran}
\IEEEoverridecommandlockouts
\usepackage{cite}
\usepackage{float}
\usepackage{amsmath,amssymb,amsfonts}
\usepackage{algorithmic}
\usepackage{graphicx}
\usepackage{textcomp}
\usepackage{xcolor}
\usepackage{amsthm}
\def\BibTeX{{\rm B\kern-.05em{\sc i\kern-.025em b}\kern-.08em
    T\kern-.1667em\lower.7ex\hbox{E}\kern-.125emX}}
\begin{document}
\title{Channel Inclusion Beyond Discrete Memoryless Channels
\author{Cihan Tepedelenlioglu   \;\; School of ECEE   \;\; Arizona State University}
}
\maketitle
\newcommand{\D}{\mathbf{D}}
\newcommand{\A}{\mathbf{A}}
\newtheorem{thm}{Theorem}
\def\bK{{\bf K}}
\def\ti{\Theta_{\rm i}}
\def\to{\Theta_{\rm o}}
\def\bB{{\bf U}_{\rm i}}
\def\bY{{\bf Y}}
\def\bH{{\bf H}}
\def\bV{{\bf V}}
\def\bX{{\bf X}}
\def\bC{{\bf U}_{\rm o}}
\def\bU{{\bf U}}
\def\ui{{\bf U}_{\rm i}}
\def\uo{{\bf U}_{\rm o}}
\def\bV{{\bf V}}
\def\bZ{{\bf Z}}
\def\bI{{\bf I}}
\def\bK{{\bf K}}
\def\bR{{\bf R}}
\def\bT{{\bf T}}
\def\a{\alpha}
\def\cX{{\cal X}}
\def\cY{{\cal Y}}
\def\by{{\bf Y}}
\def\bW{{\bf W}}
\def\bH{{\bf H}}
\def\bx{{\bf X}}
\def\bv{{\bf V}}
\def\bS{{\bf \Sigma}}
\def\bL{{\bf \Lambda}}
\def\inc{\subseteq}
\renewcommand \baselinestretch{.99}
\begin{abstract}
Partial ordering of communication channels has applications in performance analysis, and goes beyond comparisons of channels just on the basis of their Shannon capacity or error probability.  Shannon defined a partial order of channel inclusion based on convex mixture of input/output degradations of a discrete memoryless channel (DMC). In this paper, extensions to channels other than DMCs are considered. In particular, additive noise channels and phase degraded channels, and  multiple input multiple output (MIMO) linear Gaussian channels (LGC) are considered. For each of these models, the conditions under which the partial order becomes a lattice are also discussed.
%Finally, linear filtering channels with additive Gaussian noise is also addressed.
\end{abstract}

\section{Introduction}
Partial orders between communication channels have been introduced in the literature \cite{partial,shannonlattice,nasser1,nasser2,rag,polyanskiy,yuan}, and have applications in single and multiuser performance bounds. These orders reveal more structure than comparing channels using their Shannon capacities, which provides a total order on communication channels. {\it Channel inclusion} is a partial order over discrete memoryless channels (DMC) proposed by Shannon in \cite{partial}. According to channel inclusion, the ``worse (included) channel'' is a convex mixture of input/output degraded versions of the ``better (including) channel'' (please see also Figure 1). Shannon showed that if two channels sharing the same message set are ordered through inclusion, then for every block length, and every code and decoding rule for the worse channel, there exists a deterministic code and decoding rule for the better channel such that its average error probability is at most as that of the code for the worse channel. This latter order, termed ``{\it better in the Shannon sense}'' in \cite[pp. 99]{ck}, is implied by channel inclusion, but different from it.
Shannon also studied how inclusion is affected by ``sum" and ``product" of communication channels, and commented on when such a partial order has lattice structure in the sense that greatest lower bounds and least upper bounds of pairs of channels exist within this partial order.
The consequences of such a lattice structure between information elements were further elaborated  upon in  \cite{shannonlattice}. Raginsky in \cite{rag} generalized this partial order to arbitrary alphabets and introduced the notion of ``Shannon deficiency" between channels, extending the notion of deficiency between two statistical experiments introduced by Le Cam \cite{torgersen}. Topological characterizations of Shannon ordering was considered in \cite{nasser1}. In \cite{yuan}, convex optimization approaches are used to check if two DMCs are ordered. Reference \cite{cuff} considered, without making an explicit connection to channel inclusion, an approximate asymptotic version of channel inclusion in the context of channel synthesis or simulation.

%The study of Shannon inclusion have been introduced for DMCs, and have  been extended to arbitrary input/output alphabets, in \cite{rag}. However,
Efforts to concretely establish orders between  communication channel models, such as seen in wireless communications are lacking in the literature. Such channels often require input constraints, and require the definition of inclusion to involve an input degradation to be restricted to a certain class.  This also leads to the idea of restricting both the input {\it and} output degradations to a certain class, rather than considering all possible degradations.
%For example, partial ordering of linear Gaussian channels (LGC) can be considered by degradations that are only linear, and input constraints can be handled by norm constraints on input degradation matrices.

With this background, the contributions in this paper as follows:
\begin{itemize}
\item As an example of output-only degradation additive noise channels are considered, and it is shown that if the noise is from an infinitely divisible distribution, this order constitutes a lattice.
\item Input/output phase degradation of channels that are probability distributions over the torus are considered. The ordering of these channels are studied in terms of their two-dimensional (characteristic function) Fourier series coefficients.
\item Linear Gaussian channels with deterministic and random channel coefficients are considered, where the input degradation matrices have norm constraints.
\end{itemize}

In what follows, channel inclusion will be reviewed. In Section III additive noise channels, and Section IV phase-degradation is considered. Section V studies LGCs, and in Section VI, the paper is concluded with a discussion of future work.
%\item Input/output degradations of linear filtering channels with Gaussian noise is considered.
%In the next Section, we recall the definition of channel inclusion for DMCs, and extend it to channels that are not necessarily DMCs.
%We then apply this extension to linear Gaussian channels in Section \ref{LGC}, which are MIMO channels with deterministic channel coefficients.
%\cite{cuff} \cite{dahl} \cite{torgersen} \cite{mardia}
%Information theory of angular data was considered in \cite{har} where rate-distortion problems are considered
\section{Channel Inclusion and Extensions}
%\subsection{Review of Channel Inclusion for DMCs}
%Toward this goal, we first recall the inclusion partial order, and than discuss its extensions to general alphabets. In the next section, we will apply this general extension to LGC, and characterize this specialization of the Shannon inclusion to LGCs.
Given two row stochastic matrices, we say that  $\bK_1$  ``includes"  $\bK_2$  if we can write
\begin{eqnarray}
\bK_2 = \sum_{\a \in {\cal A}} g_{\a} \bR_{\a} \bK_1 \bT_{\a} \;,
\label{inc}
\end{eqnarray}
where $\bR_{\a}$ and $\bT_{\a}$ are arbitrary row stochastic matrices representing the input and output degradation channels respectively, and $g_{\a}$ are probabilities over some index set $\alpha \in {\cal A}$. Shannon has showed that for $\bR_{\a}$ and $\bT_{\a}$ one can without loss of generality consider matrices that contain only 0s and 1s rather than the set of all stochastic matrices, and obtain exactly the same partial order. This also means that the sum in (\ref{inc}) is over finitely many terms, without loss of generality. We will use $\bK_2 \inc \bK_1$
to denote this partial order. Note that a pair of channels might not be comparable under this order. Also, the ordered channels do not have to have the same number of inputs or same number of outputs since the stochastic matrices $\bR_{\a}$ and $\bT_{\a}$ are not necessarily square. Equivalence of two channels is defined as  $\bK_2 \equiv \bK_1$ if $\bK_2 \inc \bK_1$ and  $\bK_1 \inc \bK_2$. Note that two channels with different stochastic matrices can be equivalent (consider, for example, row or column permutations of a stochastic matrix).  Since the definition of a partial order requires $\bK_2 \inc \bK_1$ and  $\bK_1 \inc \bK_2$ $\Rightarrow$ $\bK_2 \equiv \bK_1$ we group all channels into equivalence classes and identify with each channel the equivalence class rather than the corresponding stochastic matrix.
\begin{figure}
\centering
\includegraphics[width=8.5 cm, height=4 cm]{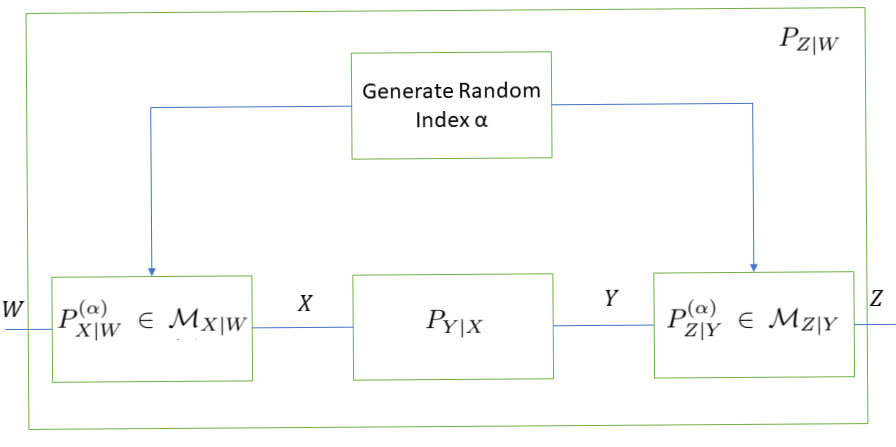}
\label{mainfig2}
\caption{Channel inclusion involves input/output degradation of a channel from a set of channels with shared randomness. The set of channels included by $P_{Y|X}$ is parameterized by the allowed input-degrading set ${\cal M}_{X|W}$ and the output-degrading set ${\cal M}_{X|W}$. Here we have $ P_{Z|W} \inc P_{Y|X}$.}
\vspace{-.3 cm}
\end{figure}
%For example, the set of all permutation matrices constitute the equivalence class of identity channels.
%\subsection{Extension to General Channels and Input Constraints} \label{extension}

The definition of inclusion can be interpreted as the worse channel being a convex mixture of input/output degraded or ``processed" version of the better channel. This interpretation can be used to extend the inclusion order to channels with continuous inputs/outputs. Let $P_{Y|X}$ be a conditional probability distribution denoting a channel with input random variable $X$ that has  range in $\cX$ to an output random variable $Y$ with range $\cY$. For discrete memoryless channels
$P_{Y|X}$ can be described by a conditional probability mass function $P_{Y|X}(y|x)$, and similarly, when the input/output can take a continuum of values the channel can be described by a conditional probability density function (PDF).
We will use the juxtaposition to denote series concatenation (composition) of two channels (which is stochastic matrix multiplication in the case of DMCs). The composition of two channels is an associative binary operation that is not necessarily commutative (as in the case of stochastic matrix multiplication).  Let  $P^{(\a)}_{Z|Y} \in {\cal M}_{Z|Y}$ and  $P^{(\a)}_{X|W} \in {\cal M}_{X|W}$ be channels indexed by $\alpha$ from pre-determined sets ${\cal M}_{Z|Y}$ and ${\cal M}_{X|W}$ of output degradation and input degradation channels, respectively, and $\alpha$ is an index over a set ${\cal A}$ denoting which input/output degradation pair is used. Let $A$ be a random variable over this index set. Then channel inclusion can be written in its most general form as
\begin{eqnarray}
P_{Z|W} = E\left[ P^{(A)}_{Z|Y}  P_{Y|X}  P^{(A)}_{X|W}\right]
\label{incgeneral}
\end{eqnarray}
where the expectation is over the distribution of the random variable $A$ with range over the index set ${\cal A}$. In this setup, the better channel is  $P_{Y|X}$,  and the worse channel is $P_{Z|W}$, so we have
$P_{Z|W} \inc P_{Y|X}$. In (\ref{incgeneral}) the expectation is an extension of the sum against $g_{\a}$ in (\ref{inc}). Note that if the random variable $A$ is a single point mass with probability 1, then the order amounts to degrading at input and output.

%\subsection{Channel Functionals that are Monotonic with Inclusion}

Since a channel including another implies it is better in the Shannon sense, the  achievable error  probability over independent uses of the channel with the best code for any block length is smaller for the including channel than that of the included channel. This means that the channel functional ``best achievable error probability'' for any message set is monotonically decreasing with the inclusion partial order for any block length. This also implies that functionals such as the channel capacity or the error exponent is increasing with channel inclusion.
\section{Additive Noise Channels}
The simplest class of channels beyond DMCs are additive noise channels.
Consider a single use of a channel with input-output relationship given by
$Y=X+V$,
%\begin{equation}
%Y=X+V \;,
%\end{equation}
where $X$ is the input, $Y$ is the output, and $V$ is additive noise. The channel $P_{Y|X}$ is completely described by the distribution of $V$. A natural question in this context is how additive noise channels can be ordered. This issue was addressed in \cite{rag} for the case of output degradation only (also known as Blackwell ordering) where it is mentioned that if a noise distribution has another one as a convolution factor, then these channels are ordered. The ordering described here can be viewed in the context of Figure 1, where ${\cal M}_{X|W}$ is empty (no input degradation), and ${\cal M}_{Z|Y}$ is the set of additive noise channels. Additive noise channels constitute partial order
if grouped in equivalence class of noise distributions which are constant translations of the distributions. The variance of the noise is a strictly increasing functional of this ordering in the sense that if two channels are ordered and not equal, then their variances will also be strictly ordered on the positive real line.
\subsection{Lattice of Infinitely Divisible Additive Noise Channels}
A lattice is a partial order where any two elements have a least upperbound, and also a greatest lowerbound within the partial order. Shannon showed in \cite{partial} that binary DMCs constitute a lattice, but it is unknown if this extends to larger numbers of inputs and outputs. In what follows, we discuss what kinds of restrictions on the noise class yields the lattice property for additive noise channels.

Consider as an example the set of all zero-mean continuous uniform distributions. Since a uniform density  cannot be decomposed into a convolution of two densities, no pair in this set of additive noise distributions can be ordered. In order theory, such a partial order is called an {\it antichain}.
On the other extreme, for the case of additive Gaussian noise channels, the partial order of having a convolution factor constitutes a total order (any two channels are ordered). When the additive noise is Gaussian, and ${\cal M}_{Z|Y}$ is the set of additive white Gaussian noise (AWGN) channels, the ordering reduces to comparisons of variances of Gaussian distributions. This total order is also a lattice where the least upperbound is the distribution with the larger variance, and the greatest lower bound is the distribution with the smaller variance.

For a more interesting example, consider the set of zero-mean infinitely divisible noise distributions with finite variance for the additive noise channels that are being ordered, as well as for the set ${\cal M}_{Z|Y}$. A distribution is infinitely divisible if the associated random variable (RV) can be written  as a sum of $n$ i.i.d. random variables, for every $n$ (see e.g., \cite{lukacs}). A fundamental representation theorem asserts that the logarithm of the characteristic function for infinitely divisible random variables with zero-mean and finite variance is of the form
\begin{equation}
\log(\phi(j\zeta)) = \int_{-\infty}^\infty \left(e^{j\zeta u} -1 - j \zeta u \right) \frac{dK(u)}{u^2} \;,
\label{rep}
\end{equation}
where $K(\cdot)$ is a non-decreasing and bounded function such that $K(-\infty) = 0$, so that $K(\cdot)$ has the properties of a positively scaled cumulative distribution function. The integrand in (\ref{rep}) is defined for $u=0$ to be $-\zeta^2/2$.
Clearly, the sum of two independent infinitely divisible random variables is also infinitely divisible, and the representation of the sum is of the form (\ref{rep}) with the $K$ function for the sum being the sum of the $K$ functions of the random variables. Then checking if two distributions with corresponding functions $K_1(u)$ and $K_2(u)$ are ordered, amounts to verifying that the difference $K_{\rm d}(u):=K_1(u)-K_2(u)$ is a non-decreasing and bounded function. This nondecreasing property can be re-expressed as $dK_{\rm d}(u)/du \geq 0$ where the function is differentiable, and when the function has jump discontinuities they are positive.

We now proceed to verify that additive infinitely divisible noise channels form a lattice. To find the least upperbound of two distributions with corresponding $K_1(u)$ and $K_2(u)$ that are absolutely continuous, we have
\begin{equation}
\frac{dK_{\rm lub}(u)}{du} = \max\left(\frac{dK_1(u)}{du},\frac{dK_2(u)}{du}\right) \;,
\label{dirac}
\end{equation}
which, along with $K_{\rm lub}(-\infty)=0$ is the smallest $K(\cdot)$ function that is bigger than both
$K_1$ and $K_2$.
When $K_1(u)$ and $K_2(u)$ have jump discontinuities in addition to an absolutely continuous part, then the derivatives of the jump discontinuities can be captured by a Dirac delta function. The maximum of two Dirac delta functions in different locations will be their sum, and the maximum of two Dirac delta functions in the same location will be another delta function with area that is the maximum of the areas of the two.
The dual operation of the greatest lower bound is the same except with a minimum rather than a maximum in the above description.

So far we have been implicitly assuming that the noise samples are independent across time so that the comparison of channels amounts to comparing univariate distributions. For the case of AWGN, we saw that this amounts to a total order. If the noise samples are Gaussian process, not necessarily   independent, then comparison of channels amounts to point-wise comparison of their power spectral densities, because the spectral densities add when independent Gaussian processes are summed. Mathematically, the operations of comparing channels, or finding least upper / greatest lower bounds   is identical to that of comparison of univariate infinitely divisible distributions. The correspondence is that the spectral distribution function and $K(\cdot)$ are analogous (positive, monotonically increasing, and bounded), and the spectral density with the derivative of $K(\cdot)$.
\vspace{-0.2 cm}
\section{Phase-Degradation of Channels}
\vspace{-0.2 cm}
For additive noise channels, we only allowed output-degradation. This is because additive noise applied at the input side can be combined with the output degradation due to the fact that addition is commutative. So even if shared randomness is allowed, as in Figure 1, the input/output degradation action of independent additive noise can be equivalently reduced to an output-only degradation. This is not possible in DMCs due to the fact that stochastic matrix multiplication is non-commutative. We now investigate another example where the input and output degradations actions are not commutative.
Consider a single use of a channel given by
\vspace{-0.1 cm}
\begin{eqnarray}
\vspace{-0.1 cm}
Y = |H| e^{j\Theta_H} X + |V|e^{j\Theta_V}
\label{phase}
\end{eqnarray}
where $X$ is the input, $H=|H| e^{j\Theta_H}$ is a complex-valued fading channel coefficient, and
$V=|V|e^{j\Theta_V}$ is complex-valued noise.
\begin{figure}
\vspace{-1.1 cm}
\centering
\includegraphics[width=9 cm, height=8.0 cm]{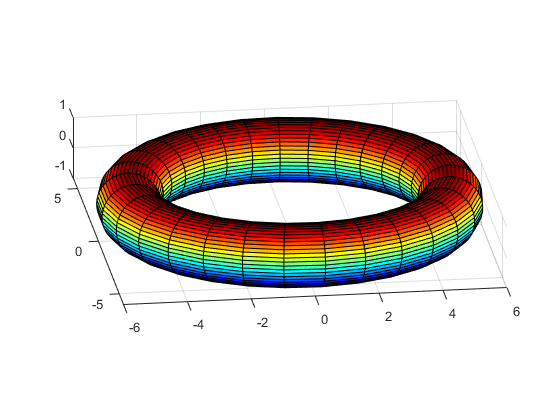}
\vspace{-1.9 cm}
\caption{The joint phase distribution of the channel and the noise $(\Theta_H,\Theta_V)$ represents a random variable on the torus. The vertical slices represent $\Theta_H$, and the horizontal slices correspond to $\Theta_{V}$ values. Input/output phase degradation $(\ti,\to)$ can also be viewed as a  distribution over the torus. }
\vspace{-.6 cm}
\end{figure}
For simplicity, we begin by assuming that $|H|$ and $|V|$ are deterministic constants. In that case, the probability distribution of the channels is characterized by the joint PDF $P_{\Theta_H\Theta_V}(\theta_H,\theta_V) \;$ with $(\theta_H,\theta_V) \in [0,2\pi)^2$, which is a probability distribution over the torus (please see Fig. 2). Since the angle random variables have bounded support, the joint PDF can be expanded using a two dimensional Fourier series (FS) which is the characteristic function:
\begin{equation}
\small
P_{\Theta_H\Theta_V}(\theta_H,\theta_V) =  \sum_{m = -\infty }^\infty \sum_{n=-\infty}^\infty \frac{\phi_{H,V}[m,n] }{(2\pi)^2}e^{-j (m \theta_H + n \theta_V)} 
\end{equation}
with the two-dimensional complex-valued FS coefficients given by
\begin{equation}
\phi_{H,V}[m,n] = E[e^{j(m\Theta_H + n \Theta_V)}] \;. \label{FS}
\end{equation}
The class of degradations that we will consider for this channel are random phase shifts from the input (multiplying $X$ with $\exp(j\Theta_{\rm i}$)), or output (multiplying $Y$ with $\exp(j\Theta_{\rm o}$)), with the pair of phases $(\ti, \to)$ (a random point on the torus) with an arbitrary joint distribution, independent of $(\Theta_H,\Theta_V)$. So the degraded channel has an input/output relation
\begin{eqnarray}
Z = |H| e^{j(\Theta_H+\to+\ti)} W + |V|e^{j(\Theta_V+\to)} \;.
\label{phase2}
\end{eqnarray}
%which shows that the phase degradation of the channel is equivalent to the channel and noise phases being degraded by phase shifts that are arbitrarily correlated.
In the FS domain, this amounts to multiplying $\phi_{H,V}[m,n]$ pointwise with another set $d[m,n]$ of FS for the joint distribution of $(\to+\ti,\to)$ which is two-dimensional circular convolution of bivariate PDFs.

Because the output degradation affects both the channel and noise phases, whereas the input-degradation only effects the signal phase, these operations cannot be combined into one phase output-degradation operation, unlike the additive noise case. Also, unlike the additive noise case, there is a ``worst channel'', where both the signal and noise phases are uniformly distributed and independent of one another (uniform distribution over the torus) with a corresponding set of FS coefficients $\phi_{H,V}[m,n] = \delta[n]\delta[m]$, where $\delta[\cdot]$ is the Kronecker delta. Moreover, using (\ref{phase2}) it is easy to see that for a fixed channel phase $\Theta_H$ the most degraded noise phase distribution is uniform (this can be done by setting $\to=-\ti$ and selecting $\to$ to be  uniform).
The corresponding FS coefficients are given by $\phi_{H,V}[m,n] = E[e^{j m \Theta_H}] \delta[n]$.
Also, for a fixed noise phase distribution, the worst channel phase distribution is also uniform, which is obtained by input-degrading with a uniform phase. The corresponding FS coefficients are given by $\phi_{H,V}[m,n] = \delta[m]  E[e^{j n \Theta_V}]$. Similar results about extremal channels can be derived using the mutual information. However, since channel inclusion is not implied by ordering mutual informations, the results mentioned here are stronger.

Channels where the phase of the noise is uniform and the channel phase is arbitrary can be described by the characteristic function (FS coefficients) $\phi_V[n]$. Similar to the additive noise case, restricting to phase distributions that are infinitely divisible on the circle we can obtain a lattice. In this case, $\phi_V[m]=\phi^k[m]$ for every $k \in \mathbb{N}$ for some characteristic function sequence $\phi[n]$. For example, starting with an infinitely divisible distribution on the real line, and wrapping it on the circle one obtains an infinitely divisible distribution on the circle whose characteristic function FS coefficients are given by samples of the characteristic function of the original continuous distribution. Examples include wrapped Gaussian, or wrapped Cauchy distributions \cite{mardia}.

%In the case of additive noise channels, it was clear that the degradation through additive noise could not be undone with adding another independent noise sample. This is another way of saying that if we identify the channels with the noise distributions, then additive noise channels constitute a partial order. This was not the case for DMCs where two DMCs with  different stochastic matrix representations could be equivalent, which required the need for equivalence classes of stochastic matrices being identified with a single channel.
Note that a channel described by the pair of random variables $(\Theta_H, \Theta_V)$ is clearly equivalent to one described by $(\Theta_H + \gamma_1, \Theta_V + \gamma_2 )$ where $\gamma_i$ are deterministic angles.
%Note that since the joint PDF is nonnegative, the FS coefficients are not arbitrary, and in particular, they satisfy $|\phi_{H,V}[m,n]| \leq 1$ with $\phi_{H,V}[0,0]=1$.
More generally, we will define phase degradations ``strict'', if they cannot be undone with a subsequent phase degradation.
In fact we can prove the following theorem about phase degradations that are strict:
\begin{thm}
Suppose the channel in (\ref{phase}) is degraded into (\ref{phase2}) through input/output phase random variables $(\ti, \to)$, and $\phi_{H,V}[m,n]  \neq  \delta[n]\delta[m]$ (i.e., the channel and noise phases are not uniformly distributed over the torus). If for any $a,b,k \in \mathbb{Z}$, and $\gamma \in [0,2\pi)$ the random variables do not satisfy
$ a \ti +b  \to = 2 \pi k+ \gamma $ with probability one, then this degradation cannot be undone with another phase-degradation operation.
\label{thm0}
\end{thm}
\begin{proof}
%Since the joint PDF is nonnegative, the FS coefficients are not arbitrary, and in particular, they satisfy $|\phi_{H,V}[m,n]| \leq 1$ with $\phi_{H,V}[0,0]=1$.
Since $\phi_{H,V}[m,n]  \neq  \delta[n]\delta[m]$, there exists $n,m \in \mathbb{Z}$ such that $|\phi_{H,V}[m,n]|  >0$. Recall that the phase degradation operation is simply an element-wise multiplication of the two-dimensional FS coefficients. Also, the magnitude of the FS coefficients (characteristic function) are less than or equal to one. So two successive phase-degradation operations that undo one another must have a magnitude of one, which means that the FS coefficients $d[m,n]$ of the joint distribution of $(\to+\ti,\to)$ satisfies $|d[m,n]|=1$. Using this, we can write
$d[m,n] = e^{j\gamma}$, or $E[e^{j(m(\to+\ti) + n\to -\gamma}]=1$, which implies $E[\cos((a\ti + b\to -\gamma)]=1$ for integers $a=m$ and $b=m+n$. But since $\cos(x) \leq 1$, we must have $\cos((a\ti + b\to -\gamma)=1$ almost surely. But this contradicts with the assumption in the theorem.
\end{proof}
\vspace{-.2 cm}
So far we have assumed that the channel and noise magnitudes are fixed values. In general, these could be random variables, or indeed, stochastic processes across time. In addition to degrading of the phases, one refinement of (\ref{phase2}) is $Z = |H| |U| e^{j(\Theta_H+\to+\ti)} W + |V|e^{j(\Theta_V+\to)}$, where $|U| \leq 1$ is a positive random variable with an arbitrary distribution, and all the random variables can be a function of a time index. Intuitively the reason for the magnitude constraint on $U$ is to ensure that the input degradation does not improve the SNR; mathematically, this is done by restricting the class of input degradations.
For channels with input constraints, the set of possible input-degrading channels ${\cal M}_{X|W}$ can be appropriately chosen to respect an input constraint.
To expand on this idea, we consider a more general setup where the channels and the degradation operations are matrices.
\vspace{-.1 cm}
\section{Linear Gaussian Channels}
\vspace{-.1 cm}
\label{LGC}
In this section we focus on LGCs where the input/output processors are restricted to linear operators (matrices). So the channel juxtaposition  is a linear operation.
Consider a Gaussian multiple input multiple output linear channel model described by
%\begin{eqnarray}
$\bY = \bH \bX + \bV$
%\label{lg}
%\end{eqnarray}
where $\bX$ is the input vector, $\bY$ is the output vector, $\bH$ is a (not necessarily square or full-rank) matrix with deterministic coefficients, and $\bV$ is an additive Gaussian noise vector with covariance matrix $\bS$. We will assume that all vectors and matrices are real valued, keeping in mind that extensions to the complex case are rather straightforward. The channel $P_{\by|\bx}$ can be described by a multivariate Gaussian distribution ${\cal N}(\bH\bx,\bS)$. We will assume for convenience that the covariance matrix is full rank, and associate with the linear Gaussian channel the pair of matrices $(\bH,\bS)$.
%Note that the characterization of one linear Gaussian channel through a pair of matrices is different from the corresponding characterization of DMCs through stochastic matrices, even though processing of inputs/outputs of these channels can be expressed through algebraic operations in both cases.
Consider an arbitrary output processing matrix $\bC$, and input  processing matrix restricted to matrices having bounded operator norm.
So we can write $(\bC\bH\bB,\bC\bS\bC^T) \inc (\bH,\bS)$, where $\|\bB\| \leq 1$. We overload the notation $\inc$ to the matrix pair (input transformation, covariance matrix) that describes the linear Gaussian channel.
%Just like in the DMC case, we will use $\equiv$ to denote channel equivalence.
We have the following:.
\begin{thm}
If all the singular values of $\bB$ are 1,  $\bB$ is right-invertible, and $\bC$ is left-invertible, then
\begin{equation}
(\bC\bH\bB,\bC\bS\bC^T)  \equiv  (\bH,\bS). \label{lgcinc}
\end{equation}
\label{thm1}
\end{thm}
\vspace{-1.1 cm}
\begin{proof}
%$(\bC\bH\bB,\bC\bS\bC^T) \inc (\bH,\bS)$
That the right hand side includes the left hand side in (\ref{lgcinc}) follows by definition.
To show
$(\bH,\bS) \inc (\bC\bH\bB,\bC\bS\bC^T) $ one can use pseudo-inverses $\bC^\dagger$, and $\bB^\dagger$ to get $(\bC^\dagger \bC\bH\bB \bB^\dagger,\bC^\dagger\bC\bS\bC^T\bC^{\dagger T})
$, which will yield the original channel $(\bH,\bS)$. Because the singular values of $\bB$ are 1, we also have  $\|\bB^\dagger\| = 1$ so that $\bB$ is a valid input degradation.
%The proof uses the fact that one can undo the input and output processing due to the invertibility.
\end{proof}
Note that it is possible that the degrading input or output matrices do not have full rank while retaining equivalence, if the original channel matrix $\bH$ also does not have full rank. Consider, for example the extreme case where $\bH={\bf 0}$.
%As a partial converse, one can assert that if (\ref{lgcinc}) holds for every channel $(\bH, \bS)$, then all the singular values of $\bB$ are 1,  and both $\bB$, $\bC$ are right-invertible and left-invertible, respectively.

We will now show that for two linear Gaussian channels to be ordered with the above-described restriction of channel inclusion, the sorted singular values of the equivalent whitened channels will be element-wise ordered. This intuitively corresponds to sorting the signal to noise ratios and element-wise comparing them. Since the covariance matrix is full rank, we can use  Theorem 2 and set $\bC = \bS^{-1/2}$, $\bB={\bf I}$, we establish that every channel is equivalent to one with white noise:
$(\bH,\bS) \equiv (\bS^{-1/2}\bH,\bI)$ where $\bI$ is the identity matrix. Performing a singular value decomposition (SVD) to the whitened channel matrix
$\bS^{-1/2}\bH$ and using the resulting orthogonal matrices in the input and output degradation, one can convert any deterministic channel into a matrix
$\bL$ with the same dimensions as $\bH$ with a diagonal matrix with sorted singular values along its diagonal embedded in it, as in the classical SVD.
This shows that comparison of two LGCs amounts to comparing two sets of deterministic sorted singular values:
\begin{thm}
Two linear Gaussian channels are ordered if and only if the sorted singular values of their equivalent channel matrices are element-wise ordered. Moreover, this partial order is a lattice in the sense that any two channels has a greatest lower bound given by the element-wise minimum of the two positive vectors,  and a least upper bound given by the element-wise maximum.
\end{thm}
%Note that the sorting of the singular values is necessary since channel equivalence is not affected by permutation matrices.
Note also that two channels under this ordering might not be comparable. For example, for a $2\times 2$ MIMO channel defining $\sigma_i^{(k)}$ to be the $k^{{\rm th}}$ sorted singular value of the $i^{{\rm th}}$ channel, it is possible that  $\sigma_1^{(1)} > \sigma_2^{(1)}$ and  $\sigma_1^{(2)} < \sigma_1^{(2)}$. So the SNR of the first ``stream" is better for the first channel compared to the second channel, but the situation is reversed for the SNR of the second stream for this pair of unordered channels.

%The channel transformations and representations of MIMO channels in this form is quite common in the wireless communication literature. Our contribution here in this section is to show that these operations can be understood within the abstract channel inclusion framework.
\subsection{Random Channel Coefficients}

In ordering random matrix channels with AWGN, the output degradations will be assumed to be unitary to maintain the whiteness of the noise, and the input degradations will be assumed to satisfy operator norm less than one, like before. The difference will be that they will be chosen randomly with shared randomness as depicted in Figure 1.
Consider  $\bH$ having random coefficients and $\bV$ being white noise with identity covariance matrix so that
$$P_{\bY|\bX} = C E[\exp(-\| \bY-\bH\bX\|^2/2),$$
where the expectation is with respect to $\bH$, and $C$ is a normalizing constant for the PDF.
Consider a pair of random matrices $(\ui, \uo)$ so that $\|\ui\| \leq 1$ with an otherwise arbitrary joint distribution so that
$\bX=\ui \bW$, and
$\bZ = \uo \bY$. Then we have
$$P_{\bZ|\bW} = C E[\exp(-\| \bZ-\uo\bH\ui\bW\|^2/2),$$
where the expectation is with respect to $\bH$, and the random matrices $(\ui, \uo)$. Note that the norm constraint on $\ui$ rules out SNR-improving degradations, which is one way to restrict the degrading set ${\cal M}_{X|Y}$.

We distinguish two cases. If the random channel matrix $\bH$  is available at the transmitter and the receiver, then matrices $(\ui, \uo)$ can be selected as the orthogonal matrices in the SVD of $\bH$. In this case every channel matrix is fully characterized by the joint distribution of its singular values. Moreover, since the input degradation matrix can have norm less than 1, the ordering of two such channels means that the singular values of one channel is ``usual multivariate stochastically ordered'' with respect to another \cite[pp. 66]{shaked}. In particular, this means that the expected value of any coordinate-wise nondecreasing function of the singular values are ordered.
If the distribution function of the singular values for the two channels share a common copula, which captures the dependence structure, then this multivariate stochastic order
is a lattice \cite{scarsini}.

%In this setup, channels are not uniquely identified with the random matrix $\bH$ since deterministic orthogonal actions $\ui$ and $\uo$ at the input and output, respectively, would create channel coefficients $\uo \bH \ui$ creates a new channel matrix distribution, even though it is equivalent to $\bH$ since the unitary matrices can be undone with their inverses.
%Even when $\ui$ and $\uo$ are random they might not cause any degradation, if for example, they are inverses of each other,  when $\bH=\bI$ with probability 1.
If the channel is unknown at the transmitter and receiver then input/output degradations with random unitary matrices can cause strict degradations, unlike the known channel case. In this case, the distribution of the singular values of the channel does not determine the channel. In fact, for a given distribution of the singular values, extremal channels in this case can be obtained by the Haar measure on the input and output degrading matrices. For orthogonal matrices, such a ``uniform'' distribution can be obtained by using the orthogonal part of  a QR factorization of an i.i.d. Gaussian matrix.

%The lattice property of the random matrix case can be investigated by restricting the joint distribution
%on  $\ui$ and $\uo$ to infinitely divisible distributions on orthogonal matrices???
%\section{Degradation by Linear Filtering}
%\section{Shannon Deficiency}
%\section{Extensions to Multi-user Channels}
\section{Conclusions and Future Work}
Extensions of Shannon's channel inclusion partial order beyond DMCs are considered. It is shown that restricting the degradation can make the partial order into a lattice, as seen for the case of infinitely divisible additive noise channels. Degradation of channel phase is shown to be captured by a probability distribution over the torus. Extremal channels in this context are identified. MIMO channels with AWGN  with deterministic and random channel coefficients are considered.

Future work includes extensions to linear filtering channels where the channels and the degradation operations are convolutions. This amounts to the degradation of the SNR at each frequency. However, with causality and stability restrictions imposed due to practical considerations, the problem is more challenging.
%Moreover, in the case of linear filtering channels, the convolution operation changes the length of the codeword, so proving that inclusion implies being better in the Shannon sense is not possible without further assumptions.
Another interesting direction is considering multiple access and broadcast channels. In these cases extracting a deterministic code for the better channel from a random code is not straightforward since there is more than one error probability performance metric.
Finally, explicit computation of the Shannon deficiency for the channel models considered is an interesting direction.
\newpage
\bibliographystyle{IEEEtran}
\bibliography{ordering}
\end{document}